\documentclass[a4paper,UKenglish,cleveref, autoref, thm-restate]{lipics-v2021}
\usepackage{amsfonts, amsmath, amssymb, amsthm, algorithm,algpseudocode, enumitem, xspace, graphicx, float}

\PassOptionsToPackage{obeyspaces}{url}
\DeclareMathOperator{\mc}{mincut}
\DeclareMathOperator{\cl}{cl}

\newtheorem{problem}[theorem]{Problem}
\newtheorem*{lemma*}{Lemma}
\newtheorem*{thm:skew}{Theorem \ref{thm:skew}}

\newcommand*{\MA}{\mathcal{A}}
\newcommand*{\MB}{\mathcal{B}}
\newcommand*{\MF}{\mathcal{F}} 
\newcommand*{\FF}{\mathbb{F}}

\newcommand*{\ol}{\overline}   
\newcommand*{\sm}{\setminus}                           
\newcommand*{\abs}[1]{\lvert #1\rvert}

\newcommand{\Bol}{Bollob\'as}
\newcommand{\KW}{Kratsch and Wahlstr\"om\xspace}

\hideLIPIcs
\nolinenumbers

\title{Near-linear-time, Optimal Vertex Cut Sparsifiers in Directed Acyclic Graphs}

\author{Zhiyang He}{Computer Science Department, Carnegie Mellon University, Pittsburgh, United States}{}{}{}
\author{Jason Li}{Computer Science Department, Carnegie Mellon University, Pittsburgh, United States}{}{}{Supported in part by NSF award CCF-1907820.}
\author{Magnus Wahlstr\"om}{Royal Holloway, University of London, UK}{}{}{}

\authorrunning{Z.\,He and J.\,Li and M.\,Wahlstr\"om} 

\Copyright{Zhiyang He and Jason Li and Magnus Wahlstr\"om} 



\category{} 

\relatedversion{} 

\EventEditors{Petra Mutzel, Rasmus Pagh, and Grzegorz Herman}
\EventNoEds{3}
\EventLongTitle{29th Annual European Symposium on Algorithms (ESA 2021)}
\EventShortTitle{ESA 2021}
\EventAcronym{ESA}
\EventYear{2021}
\EventDate{September 6--8, 2021}
\EventLocation{Lisbon, Portugal}
\EventLogo{}
\SeriesVolume{204}
\ArticleNo{52}

\begin{document}

\maketitle
\begin{abstract}
Let $G$ be a graph and $S, T \subseteq V(G)$ be (possibly overlapping) sets of terminals, $|S|=|T|=k$.  We are interested in computing a \emph{vertex sparsifier} for terminal cuts in $G$, i.e., a graph $H$ on a smallest possible number of vertices, where $S \cup T \subseteq V(H)$ and such that for every $A \subseteq S$ and $B \subseteq T$ the size of a minimum $(A,B)$-vertex cut is the same in $G$ as in $H$.
We assume that our graphs are unweighted and that terminals may be part of the min-cut.
In previous work, Kratsch and Wahlstr\"om (FOCS 2012/JACM 2020) used connections to matroid theory to show that a vertex sparsifier $H$ with $O(k^3)$ vertices can be computed in randomized polynomial time, even for arbitrary digraphs $G$. However, since then, no improvements on the size $O(k^3)$ have been shown.

In this paper, we draw inspiration from the renowned \emph{\Bol's Two-Families Theorem} in extremal combinatorics and introduce the use of total orderings into \KW's methods. This new perspective allows us to construct a sparsifier $H$ of $\Theta(k^2)$ vertices for the case that $G$ is a DAG.  We also show how to compute $H$ in time near-linear in the size of $G$, improving on the previous $O(n^{\omega+1})$. Furthermore, $H$ recovers the \emph{closest} min-cut in $G$ for every partition $(A,B)$, which was not previously known.  Finally, we show that a sparsifier of size $\Omega(k^2)$ is required, both for DAGs and for undirected edge cuts.  
\end{abstract}

\newpage

\section{Introduction}\label{sec:intro}

Let $G=(V,E)$ be an unweighted, directed graph, and let $S,T\subset V$ be sets of terminals, not necessarily disjoint. In the vertex sparsifier problem, our goal is to construct a smaller graph $H$, called a \emph{vertex sparsifier}, that preserves the cut structure of $S,T$ in $G$. More precisely, $H$ includes all vertices in $S,T$, and for all subsets $A\subseteq S$ and $B\subseteq T$, the size of the minimum vertex cut separating $A$ and $B$ is the same in $G$ and $H$.
Here, we allow the min-cut to contain vertices from $A$ and $B$;
for other notions of cut sparsifiers from the literature, see related work, below. 

A result of \KW proved the first bound on the size of a vertex sparsifier that is polynomial in the number of terminals. When $S,T$ have size $k$, the vertex sparsifier has $O(k^3)$ vertices. \KW's main insight is to phrase the problem in terms of constructing representative families on a certain matroid, after which they can appeal to the rich theory on representative families~\cite{monien1985find,marx2009}. Their result, also known as the \emph{cut-covering lemma} in the areas of fixed-parameter tractability and kernelization, has led to many new algorithmic developments~\cite{KratschW20,kratsch2014compression,hols2018randomized,kratsch2018randomized}. Nevertheless, despite the recent surge in applications of the cut-covering lemma, the original bound of $O(k^3)$ has yet to be improved.

In this paper, we introduce the ordered version of the representative family method, and use it to give a sparsifier on $O(k^2)$ vertices in directed acyclic graphs.  This matches lower bounds of $\Omega(k^2)$, which we present in Section~\ref{sec:lower}.  
Furthermore, unlike \KW's result, our new algorithm runs in near-linear time in the size of the graph, and preserves all \emph{closest} min-cuts between subsets of the terminals.  In fact, the latter is an important ingredient in the improved running time; see discussion below.  
We expect that covering closest min-cuts may lead to  further applications in the theory of kernelization. The central method we use is the following theorem.

\begin{theorem}
\label{thm:skew}
Suppose $\MF$ is a family of subsets of $\FF^d$ for some field $\FF$ and for all $B\in \MF$, $\abs{B} = s$. Let 
\[
\MA = \{A\subseteq \FF^d\mid \abs{A} \le r \text{ and } \exists B\in \MF \text{ s.t.\ } A\uplus B\text{ is linearly independent}\} .
\]
Fix any ordering $\sigma$ of $\MF$, namely $\MF = \{B_1, B_2, \ldots, B_n\}$, and suppose $d\ge r+s$. Then there exists $\MB\subseteq \MF$, $\MB = \{B_{i_1}, B_{i_2}, \ldots, B_{i_m}\}$ where $i_1<i_2<\ldots<i_m$, such that 
\begin{enumerate}[label = (\alph*)]
\item \label{condition:order} For all $A\in \MA$, there exists $B_{i_k}\in \MB$ where $A\uplus B_{i_k}$ is independent and for all $j\in [n], j > i_k$, $A\uplus B_j$ is dependent. Note that $B_j$ is not necessarily in $\MB$. Here $\uplus$ denotes disjoint union, which particularly implies that if $A$ and $B$ are not disjoint, then $A\uplus B$ is a multi-set and therefore dependent.
\item $m\le \binom{d}{s}$, and we can find $\MB$ algorithmically using $O(\binom{d}{s}ns^\omega + \binom{d}{s}^{\omega-1}n)$ field operations over $\FF$ (in particular, in a number of operations linear in $|\mathcal F|$).
\end{enumerate}
\end{theorem}

The condition that $A \uplus B_j$ is required to be dependent only for $j>i_k$, as opposed to $A \uplus B_j$ being dependent for every $j \neq i_k$, recalls the \emph{skew} version of \Bol's two-families theorem, proven by Frankl~\cite{Frankl82}.  

Technically, this version is equivalent to the weighted version of the representative family method shown by Fomin et al.~\cite{fomin2016efficient}.  In that version, every element $X \in \MF$ has a weight $w(X)$, and condition \ref{condition:order} is replaced by the condition that $w(B_{i_k})$ is maximal among all sets $B_j$ such that $A \uplus B_j$ is independent. Indeed, given $\sigma$ we can simply use $w(B_j)=j$, and conversely any input where all the weights are distinct enforces a corresponding total ordering on the elements\footnote{In the first version of this paper, we presented a proof of Theorem~\ref{thm:skew} that runs in polynomial time. It was later pointed out to us that this theorem is equivalent to Theorem~3.7 in \cite{fomin2016efficient}, which runs in near-linear time. We thereby refer the readers to the proof in \cite{fomin2016efficient}, as our proof shares similar underlying ideas with theirs.}.  
However, the difference in focus between weights and an ordering is significant, as a total element ordering can carry semantic meaning that is obscured when implemented using weights.

Applying Theorem~\ref{thm:skew} to vertex cut sparsifiers, we obtain the main result of this paper.

\begin{theorem}\label{thm:main}
Given a directed acyclic graph $G = (V, E)$ with terminal sets $S, T$ of size $k$, we can find a vertex cut sparsifier of $G$ of size $\Theta(k^2)$ algorithmically in time $\tilde{O}((m+n)k^{O(1)})$.
\end{theorem} 

Our proof initially follows that of \KW~\cite{KratschW20}.  Like \KW, we compute a
``cut-covering set'', i.e., a set $Z \subseteq V(G)$ such that for
every $A \subseteq S$ and $B \subseteq T$ there is an $(A,B)$-min cut
$C$ with $C \subseteq Z$.  However, \KW's result is based around
\emph{essential vertices}, i.e., vertices $v \in V(G)$ that must be
included in any cut-covering set.  Using the unordered version of
Theorem~\ref{thm:skew} (see Lemma~1.1 of~\cite{kratsch2012representative}), they compute a set $X$ of $O(k^3)$ vertices which
is guaranteed to include all essential vertices in $G$.  They then eliminate one vertex not in $X$, recompute the representative family, and repeat until exhaustion.  This gives a superlinear running
time and a final bound of $|Z|=O(k^3)$.
This iterative process for computing $Z$ was required because the initial set $X$
could not be guaranteed to contain \emph{all} vertices of any $(A,B)$-min cut $C$,
unless that min-cut happens to be unique (i.e., consist only of essential vertices). 

We use Theorem~\ref{thm:skew} to improve over this in two ways.
By using the additional power of an ordering, we reduce the bound from
$|Z|=O(k^3)$ to $|Z|=O(k^2)$ when $G$ is a DAG.  Furthermore, instead
of covering some arbitrary min-cut for every pair $(A,B)$, we observe
that our construction guarantees that $Z$ contains the unique
\emph{closest} $(A,B)$-min cut to $A$, i.e., that mincut $C$ for which
the set of vertices reachable from $A$ is minimized.  This is an
interesting consequence that was not previously known, but
additionally, it allows us to significantly improve the running time required to
compute a cut-covering set.  Since the output
of Theorem~\ref{thm:skew} contains all vertices of the closest $(A,B)$-min cut $C$ for every $(A,B)$,
we can compute a cut-covering set $Z$ with a
single application of Theorem~\ref{thm:skew}, eliminating the iterative
nature of~\cite{KratschW20} and improving the running time of the
procedure. 


Finally, to achieve a linear running time, one more obstacle must be 
overcome.  In order to apply Theorem~\ref{thm:skew}, we need to compute a representation for the matroids underlying the result, known as \emph{gammoids}, in time linear in $n+m$.  The usual method for representing gammoids goes via the class of \emph{transversal matroids}, however, this requires taking the inverse of an $n \times n$ matrix.  Luckily, we observe that an older construction of Mason~\cite{mason72gammoids} can be used to represent gammoids more efficiently over DAGs; see Lemma~\ref{lem:quick-repr}.
This allows us to compute $Z$ in time linear in $n+m$, and computing the final sparsifier $H$ is then a simple task.

We also show that there are graphs $G$ with $k$ terminals such that any cut sparsifier $H$ requires $\Omega(k^2)$ vertices, both when $H$ is a directed vertex cut sparsifier for a DAG $G$, and when $H$ is an undirected edge cut sparsifier for an undirected graph $G$. 

\subparagraph*{Related work.} 
Several different notions of cut sparsifiers have been considered, as well as vertex sparsifiers for other problems.
Vertex cut sparsifiers were first introduced by Moitra~\cite{moitra2009approximation} in the setting of approximation algorithms; see also~\cite{makarychev2010metric,charikar2010vertex,chuzhoy2012vertex}. Recently, they have found applications in fast graph algorithms, especially in the dynamic setting~\cite{durfee2019fully,chalermsook2020vertex,chen2020fast}. 
Compared to our setting, many of these results replicate
min-cuts only approximately, rather than exactly, and most apply only
to undirected edge cuts.  On the other hand, in the general case
(e.g., when working with edge cuts or when terminals are not deletable),
there is an important distinction between parameterizing by the
\emph{number of terminals} and the \emph{total terminal capacity}.
Our setting, with deletable terminals, corresponds to parameterizing
edge cuts by the total capacity of the terminal set.  Previous results
for this setting include \KW~\cite{KratschW20} discussed above
and Chuzhoy~\cite{chuzhoy2012vertex}, as well as recent results on terminal multicut sparsification~\cite{Wahlstrom20multicut}.  By contrast, parameterizing by
the number of terminals in a setting where terminals have unbounded capacity makes
for a much harder sparsification task, and this is the setting that
has been most commonly considered in approximation algorithms.
Indeed, it is known that any exact cut sparsifier for $k$ terminals with unbounded capacity needs to have at least exponential size in $k$, and possibly double-exponential~\cite{krauthgamer2013mimicking}.

\section{Preliminaries} \label{sec:prel}
Throughout the paper, all graphs are directed and unweighted. We begin with standard terminology on cuts and cut sparsifiers.
\begin{definition}[Vertex Cut]
Given a directed unweighted graph $G = (V, E)$ with sets $X, Y\subseteq V$, a set $C\subseteq V$ is a \emph{vertex cut} of $(X, Y)$ if after removing $C$ from $G$, there does not exist a path from a vertex in $X$ to a vertex in $Y$. We denote the size of a minimum vertex cut between $X, Y$ in $G$ as $\mc_G(X, Y)$. 
\end{definition}
\begin{definition}[Vertex Cut Sparsifier]
Consider a directed unweighted graph $G = (V, E)$ with sets $S, T\subseteq V$. A directed unweighted graph $H = (V', F)$ is a \emph{vertex cut sparsifier} of $G$ if 
\begin{enumerate}[label = (\alph*)]
\item $S, T\subseteq V'$.
\item For all $X\subseteq S, Y\subseteq T$, $\mc_G(X, Y) = \mc_H(X, Y)$.
\end{enumerate}
\end{definition}

The problem we consider in this paper is the minimum size of a vertex sparsifier.

\begin{problem}[Minimum Vertex Cut Sparsifier]
Given a directed unweighted graph $G = (V, E)$ with terminal sets $S, T\subseteq V$, what is the minimum number of vertices in a vertex cut sparsifier of $G$? 
\end{problem}

\KW~\cite{KratschW20} obtained a bound for this problem of $O(k^3)$ vertices, where $|S|=|T|=k$. Their application was in the fixed-parameter tractability setting, specifically in constructing kernels for cut-based problems. Our proof utilizes similar matroid-theoretic techniques as in their work. We introduce these techniques next.

\begin{definition}[Matroid]
Given a finite ground set $E$, a set system $M = (E, I)$ where $I\subseteq \mathcal{P}(E)$ is called a \emph{matroid} if
\begin{enumerate}[label = (\alph*)]
\item $\varnothing\in I$.
\item For $X, Y\subseteq E$, if $Y\in I$ and $X\subseteq Y$, then $X\in I$. 
\item If $X, Y\in I$ and $|X| < |Y|$, then there exists $y\in Y\setminus X$ such that $X\cup \{y\} \in I$. 
\end{enumerate}
\end{definition}

Central to our proof is the use of gammoids and their representations.

\begin{definition}[Gammoid]
Given a graph $G=(V,E)$ and a subset of vertices $S$ (which we refer to as the ``source vertices''), the \emph{gammoid} on $S$ is the matroid $M = (V, I)$ where $I$ contains all subsets $T\subseteq V$ such that there exist $|T|$ vertex-disjoint paths from $S$ to $T$ in $G$.
\end{definition}

\begin{definition}[Matroid disjoint union]
Given two matroids on disjoint ground sets, their \emph{matroid disjoint union} is the matroid whose ground set is the union of their ground sets, and a set is independent if the corresponding part in each matroid is independent. 
\end{definition}

\begin{definition}[Representable matroid]
Given a field $F$, a matroid $M = (E, I)$ is \emph{representable} over $F$ if there exists a matrix $A$ over the field $F$ and a bijective mapping from $E$ to the columns of $A$, such that a set $S\subseteq E$ is independent if and only if its corresponding set of columns of $A$ are linearly independent. 
\end{definition}

In particular, it is well known in matroid theory that gammoids are representable in randomized polynomial time; see Marx~\cite{marx2009}.  However, to control the running time, we revisit an older representation by Mason~\cite{mason72gammoids}, and note that it leads to a representation in near-linear time in $|V|+|E|$ on DAGs. (Proofs of results marked $\bigstar$ are deferred to the appendix.)

\begin{lemma}[$\bigstar$ Construction of Gammoid Representation on DAGs] \label{lem:quick-repr}
Given a directed acyclic graph $G = (V, E)$, a set $S\subseteq V$, and $\varepsilon > 0$, with $|V|=n$, $|E|=m$ and $|S|=k$, a representation of the gammoid on $S$ of dimension $k$ over a finite field with entries of bit length $\ell=O(k \log n + \log(1/\varepsilon))$ can be constructed in randomized time $\tilde O((n+m) k\ell)$ with one-sided error at most $\varepsilon$, where $\tilde O$ hides factors logarithmic in $\ell$. 
\end{lemma}

We also note that the disjoint union of two representable matroids is representable. Given these definitions, we now present the main arguments of this paper.



\section{Vertex Cut Sparsifier for DAGs}

In this section, we prove our main result, Theorem~\ref{thm:main}. We first borrow the following key concepts from \KW~\cite{KratschW20}.

\begin{definition}[Essential Vertex]
A vertex $v\in V\setminus(S\cup T)$ is called \emph{essential} if there exist $X\subseteq S, Y\subseteq T$ such that $v$ belongs to every minimum vertex cut between $X, Y$.
\end{definition}

\begin{definition}[Neighborhood Closure]
For a digraph $G = (V, E)$ and a vertex $v\in G$, the neighborhood closure operation is defined by removing $v$ from $G$ and adding an edge from every in-neighbor of $v$ to every out-neighbor of $v$. The new graph is denoted by $\cl_v(G)$. 
\end{definition}

\begin{definition}[Closest Set]
For sets of vertices $X, A \subseteq V$, $A$ is closest to $X$ if $A$ is the unique min-cut between $X$ and $A$. 
\end{definition}

We introduce the following definitions to simplify our discussions.
\begin{definition}\label{def:LR}
For sets $X\subseteq S, Y\subseteq T$, let $C$ be a vertex cut for $X, Y$. Let $G'$ be the subgraph formed by the union of all paths from $X$ to $Y$. The \emph{left-hand side} of $C$, denoted $L(C)$, is the set of vertices that are still reachable from $X$ in $G'$ after $C$ is removed. Similarly, the \emph{right-hand side} of $C$, denoted $R(C)$, is the set of vertices that can still reach $Y$ in $G'$ after $C$ is removed. 
\end{definition}

\begin{definition}[Saturation]
Let $C$ be a vertex cut for $X\subseteq S, Y\subseteq T$. For a vertex $v\in C$, we say that $(C, v)$ is saturated by $X$ if there exist $|C| + 1$ paths from $X$ to $C$ that are vertex disjoint except for two paths that both ends at $v$. Similarly, $(C, v)$ is saturated by $Y$ if there exist $|C| + 1$ paths from $C$ to $Y$ that are vertex disjoint except for two paths that both start at $v$.
\end{definition}

The following three lemmas follow from~\cite[Section~5.1]{KratschW20}. Their proofs, as presented in \cite{KratschW20} and Chapter~11.6 of \cite{fomin}, are included in the Appendix for completeness.
\begin{lemma}[Essential Vertex Lemma]\label{lem:ess}
Let $v$ be an essential vertex with respect to ${X\subseteq S}$, $Y\subseteq T$. Let $C$ be any minimum vertex cut between $X, Y$. Then $(C, v)$ is saturated by both $X$ and $Y$.
\end{lemma}
\begin{lemma}[Closure Lemma]\label{lem:cl}
If $v\in V\setminus(S\cup T)$ is not an essential vertex, then $\cl_v(G)$ is a vertex cut sparsifier of $G$. 
\end{lemma}

\begin{lemma}[Closest Cut Lemma]\label{lem:close}
Let $C$ be a vertex cut for $X\subseteq S, Y\subseteq T$ that is closest to $X$ (resp. $Y$), then for all vertices $v\in C$, $(C, v)$ is saturated by $X$ (resp. $Y$).
\end{lemma}
Using the saturation properties of essential vertices (as in Lemma~\ref{lem:ess}), \KW presented a construction of matroids which, combined with the unordered version of Theorem~\ref{thm:skew}, computes a set of vertices $P$ of size $O(k^3)$ that contains all the essential vertices. They then apply Lemma~\ref{lem:cl} iteratively to any one vertex not in $P$ and repeat the process.  The repetition is required since Lemma~\ref{lem:cl} may change the essential vertices of the graph.


In our proof, instead of marking only essential vertices we consider all vertices on closest min-cuts.  These have weaker saturation properties than essential vertices (as in Lemma~\ref{lem:close}), but these weaker properties suffice thanks to the ordering imposed in Theorem~\ref{thm:skew}. This ordering, when applied to the topological ordering of a DAG, allows us to mark all vertices of closest min-cuts in a single application of Theorem~\ref{thm:skew}, which eliminates the iterative nature of the previous result.  This intuition will be made clear in the following discussion. 

We now present our construction, which results in Theorem~\ref{thm:main}.

\begin{theorem}\label{thm:real}
For a directed acyclic graph $G = (V, E)$ with terminal sets $S$ and $T$, let $k = |S\cup T|$. Then there exists a set of vertices $P$ of size $O(k^2)$ such that for each pair of $X\subseteq S, Y\subseteq T$, the min-$(X, Y)$ cut that is closest to $Y$ is contained in $P$. This set can be found in time $\tilde O(nk^{2\omega-1}+mk^2)$, and a sparsifier on $P$ can then be constructed in the same asymptotic running time.
\end{theorem}
\begin{proof}
Let $G_R = (V, E_R)$ be the graph $G$ with the direction of all edges reversed. We make the following modification to our graphs $G, G_R$. For each vertex $v\in V\setminus (S\cup T)$, add a vertex $v'$ into $V$ and for each directed edge $(u, v)\in E$, add $(u, v')$ into $E$. We refer to $v'$ as the sink-only copy of $v$. Denote this new directed graph $G' = (V', E')$, we add sink-only copies of vertices to $G_R$ and obtain $G_R' = (V', E_R')$. Note that $G', G_R'$ are both acyclic. Enumerate $V$ in a reverse topological ordering, namely $V = (v_1, v_2, \cdots, v_n)$ where $v_i$ cannot reach any $v_j$ for $j > i$. 

Let $M_1 = (E_1, I_1)$ be the gammoid constructed on the graph $G$ and the set of terminals $S$ in $G$, and let $M_2 = (E_2, I_2)$ be the gammoid constructed on the graph $G_R'$ and the set of terminals $T$ in $G_R'$. To distinguish between vertices of $E_1$ and $E_2$, we label vertices in $E_1$ as $v_1, \ldots, v_n$, and elements in $E_2$ as $\ol{v_1}, \ldots, \ol{v_n}, \ol{v_1}', \ldots, \ol{v_n}'$. Note that $E_1$ does not contain sink-only copies. For any set of vertices $U\subseteq V$, denote the respective sets in $E_1$ and $E_2$ as $U_1$ and $U_2$. Let $M$ be the disjoint union of matroids $M_1$ and $M_2$, so $M$ is representable and it has rank $O(k)$. 

Observe that for any $X\subseteq S,\, Y\subseteq T$, the min-cuts between $X$ and $Y$ in $G$ are the same as in $G'$ because the vertex copies $v'$ we added to $G$ have no outgoing edges. Therefore, $G$ and $G'$ have the same set of closest cuts.  We now consider the following constructions. For a min-cut $C$ between $X, Y$ that is closest to $Y$, and a vertex $v\in C$, define
\[
A_{(C, v)} = [(S_1\sm X_1) \cup (C_1\sm \{v\})] \cup [(T_2\sm Y_2)\cup C_2].
\]
Let $\tilde\MA$ consist of $A_{(C, v)}$ for all such cut-vertex pairs. Define 
\[
\MF = \{B_{v} = \{v, \ol{v}'\}\mid v\in V\}. 
\]
We pause for a moment to explain the ideas behind these definitions. First of all, note that the algorithm in Theorem~\ref{thm:skew} only takes as input a family $\MF$. Intuitively, one should think of the sets in $\MF$ as answers to potential queries, and the sets in $\MA$ (defined in Theorem~\ref{thm:skew}) as all queries answered by $\MF$. As we will show, the family of queries $\tilde\MA$ we defined is a subset of $\MA$. If we can further show that for each $v$ in a closest min-cut $C$, $B_v$ is the unique answer that the algorithm will find for query $A_{(C,v)}$, then we know that the output of the algorithm must contain all vertices of all closest min-cuts, because all queries in $\MA\supseteq \tilde\MA$ must be answered. 

More specifically, for each query $A_{(C, v)}$, Theorem~\ref{thm:skew} promises to find the answer $B_{u}$ to $A_{(C, v)}$ (which means $A_{(C, v)}\uplus B_{u}$ is independent in $M$) that is the last answer according to the reverse topological ordering on $\MF$. This means for all $w > u$, $A_{(C, v)}\uplus B_{w}$ is dependent in $M$. Since our goal is for Theorem~\ref{thm:skew} to output a set containing all vertices on all closest min-cuts, we want to show that $B_{v}$ is the last answer in the ordering to $A_{(C,v)}$. This is precisely captured in the following claim.
\begin{claim}\label{clm:main}
For each $A = A_{(C, v)}$, $B_{v}$ is the unique set in $\MF$ such that 
\begin{itemize}
    \item $A\uplus B_{v}$ is independent in $M$, and
    \item for all $u > v$ in the reverse topological ordering of $V$, $A\uplus B_u$ is dependent in $M$.
\end{itemize}
\end{claim}
\begin{proof}
We first show that $A\uplus B_{v}$ is independent. Since $M$ is a disjoint union matroid, we need to show that $(A\cap E_1) \uplus \{v\} = (S_1\setminus X_1) \cup C_1$ is independent in $M_1$ and 
$(A\cap E_2) \uplus \{\ol{v}'\} = (T_2\setminus Y_2) \cup C_2\cup \{\ol{v}'\}$ is independent in $M_2$.

Since both $M_1$ and $M_2$ are gammoids, we need to show existence of vertex disjoint paths from $S_1$ to $(S_1\setminus X_1) \cup C_1$. First note that singleton paths can cover all vertices in $S_1\setminus X_1$. Since $C_1$ is a min-cut between $X_1$ and $Y_1$, by duality there exist vertex disjoint paths from $X_1$ to $C_1$. Therefore $(S_1\setminus X_1) \cup C_1$ is independent in $M_1$. Similarly, singleton paths can cover all vertices in $T_2\setminus Y_2$. It suffices for us to show the existence of vertex disjoint paths from $Y_2$ to $C_2 \cup \{\ol{v}'\}$. 

By Lemma~\ref{lem:close}, there exist $|C_2| + 1$ paths from $Y_2$ to $C_2$ that are vertex disjoint except for two paths that both ends at $\ol{v}$. Therefore, we can redirect one of these two paths to end at $\ol{v}'$, and we obtain $|C_2| + 1$ vertex disjoint paths from $Y_2$ to $C_2 \cup \{\ol{v}'\}$. This proves independence. 

Now fix $u > v$ in the reverse topological ordering, so that there does not exist a path from $v$ to $u$. We want to show that either $(A\cap E_1) \uplus\{u\}$ is dependent in $M_1$, or $(A\cap E_2)\uplus \{\ol{u}'\}$ is dependent in $M_2$. Consider four possible cases: 
\begin{itemize}
\item $u$ is not on any path from $X$ to $Y$. Assume for the sake of contradiction that both $(A\cap E_1) \uplus\{u\}$ and $(A\cap E_2)\uplus \{\ol{u}'\}$ are independent. Then there must exist a path from $X$ to $u$ and a path from $u$ to $Y$, which forms a path from $X$ to $Y$ through $u$ (since $G$ is acyclic), contradiction.
\item $u\in L(C)$ (see Definition~\ref{def:LR}). Then any path from $u$ to $Y$ (or from $Y$ to $u$ in $G_R$) must intersect with $C$, which means there does not exist vertex disjoint paths from $Y$ to $C\cup \{u\}$ in $G_R$. Therefore $(A\cap E_2)\uplus \{\ol{u}'\}$ is dependent.
\item $u\in R(C)$. Then any path from $X$ to $u$ must intersect $C$. Assume for the sake of contradiction that $(A\cap E_1)\uplus \{u\}$ is independent, then there exist vertex disjoint paths from $X$ to $C\sm \{v\}\cup \{u\}$, which means there is a path from $X$ to $u$ that goes through $v$. However, since $u > v$ in the topological ordering, there does not exist paths from $v$ to $u$. This is a contradiction, so $(A\cap E_1)\uplus \{u\}$ is dependent.
\item $u\in C$. Then $u\in (A\cap E_1)$, which implies $(A\cap E_1)\uplus \{u\}$ contains two copies of $u$. Therefore it is dependent. 
\end{itemize}
This completes the proof.
\end{proof}
We now apply Theorem~\ref{thm:skew} on $\MF$ with a reverse topological ordering, and we note that the family $\tilde\MA$ we defined in this proof is a subfamily of $\MA$ defined in Theorem~\ref{thm:skew}. Let $P$ be the collection of vertices that Theorem~\ref{thm:skew} finds. Then by the above claim, for each pair of $X\subseteq S, Y\subseteq T$ and their min-cut $C$ closest to $Y$, all vertices in $C$ must be in $P$. Note that this also implies that all essential vertices are in $P$. 

To construct the final sparsifier $H$ on $P$, for each vertex $u\in P$, we run a depth-first search starting at $u$ on the graph $G_u$, defined to be $G$ minus the out-edges of vertices in $P\sm\{u\}$. For each vertex $v\in P\sm\{u\}$ that is reachable from $u$ in $G_u$, we add an edge $(u,v)$ to $H$. Note that this procedure returns the same graph $H$ as the one that sequentially applies Lemma~\ref{lem:cl} on all vertices not in $P$, but achieves a shorter runtime of $O(k^2m)$. To see the equivalence, observe that in both cases, there is an edge $(u,v)$ in the final sparsifier if and only if there is a path from $u$ to $v$ in $G$ whose internal vertices are disjoint from $P$. 
We conclude that the output graph $H$ is a valid sparsifier. 

For the final running time, note that computing the gammoids takes time $\tilde O((m+n)k^2)$ by Lemma~\ref{lem:quick-repr} with $\varepsilon=1/n^{O(1)}$, and computing the representative sets takes time $\tilde O(nk^{2\omega-1})$ by taking $d=k$ and $s=2$ in Theorem~\ref{thm:skew}.
\end{proof}

We make a few remarks regarding this proof. Intuitively, the gammoid $M_2$ and its respective query $A_{(C, v)} \cap E_2 = (T_2\setminus Y_2) \cup C_2$ is used to filter out all vertices on the left-hand side of $C$, because no vertex in $L(C)$ can reach $T$ without crossing $C$. The gammoid $M_1$ and its query $A_{(C, v)} \cap E_1 = (S_1\sm X_1) \cup (C_1\sm \{v\})$, however, is different, because the query itself allows terminals in $X$ to reach the right-hand side of $C$ through $v$. If we do not impose a reverse topological ordering on the vertices and use the unordered version of Theorem~\ref{thm:skew} (Lemma~1.1 in~\cite{kratsch2012representative}), our proof will fail -- there may exists $u\in R(C)$ reachable from $v$ such that $A_{(C,v)}\uplus B_u$ is independent, and the algorithm may not find $B_v$ as the answer. By imposing the reverse topological ordering, we demand the algorithm to find the answer last in the ordering, thereby ensuring the algorithm discovering $B_v$. 

This is the critical difference between our proof and \KW's proof. In \KW's paper, they presented a construction with three matroids -- one gammoid on $G'$ and $S_1$ (note that our first gammoid is constructed on $G$ and $S_1$), one gammoid on $G_R'$ and $T_2$, and one uniform matroid of rank $k$ on $V$. They then utilized the property that essential vertices can be saturated from both sides (see Lemma~\ref{lem:ess}) and constructed queries similar to our definition of $A_{(C, v)}$. Their first (resp. second) gammoid serves to filter out $R(C)$ (resp. $L(C)$), and they use the uniform matroid to filter out vertices $u\ne v\in C$. We managed to merge their first gammoid and uniform matroid with an asymmetrical construction, thereby proving a stronger bound.

We finally remark that in our proof, we never explicitly compute the queries $A_{(C, v)}$. As mentioned previously, the algorithm in Theorem~\ref{thm:skew} only takes as input the family $\MF$, and we are simply showing that given our construction of $\MF$, $\tilde\MA$ is a subset of queries answered. If one can show that another meaningful set of queries are answered by the same $\MF$, then they can derive more properties of the output set of the algorithm (while the output set itself remains unchanged). 

We now present tight lower bound constructions in the following section.

\section{Lower Bound Constructions}\label{sec:lower}
In this section we present two constructions graphs for which any vertex cut sparsifier requires $\Omega(k^2)$ vertices.  The first construction is presented by \KW~\cite{KratschW20} (construction found in arXiv preprint only); note that the graph is a DAG.  

\begin{lemma}[$\bigstar$]\label{lemma:kw-quadratic}
Let $S$ and $T$ be two vertex sets of size $2k$. Enumerate them as 
\[
S = \{v_1, v_1', v_2, v_2', \cdots, v_k, v_k'\}, \\
T = \{u_1, u_1', u_2, u_2', \cdots, u_k, u_k'\}.
\]
For each $i, j\in [k]$, create a vertex $w_{i,j}$ and add edges from $v_i, v_i'$ to $w_{i,j}$, and from $w_{i,j}$ to $u_j, u_j'$. 
    Any vertex cut sparsifier for the resulting graph requires $\Omega(k^2)$ vertices.
\end{lemma}

\newcommand{\emincut}{\textup{emincut}}
\newcommand{\dvec}{} 

Next, we show the same bound holds also for \emph{edge cut} sparsifiers, when both $G$ and its sparsifier $H$ are undirected.  
For simplicity, we also restrict ourselves to just one terminal set $T$.

\begin{definition}[(Undirected) Edge Cut Sparsifier]
Consider an unweighted, undirected graph $G = (V, E)$ with terminal set $T\subseteq V$. An unweighted, undirected graph $H = (V', F)$ is an \emph{edge cut sparsifier} of $G$ if 
\begin{enumerate}[label = (\alph*)]
\item $T\subseteq V'$.
\item For all disjoint $X, Y\subseteq T$, $\emincut_G(X, Y) = \emincut_H(X, Y)$.
\end{enumerate}
\end{definition}



\begin{figure}\centering
\includegraphics[scale=1.2]{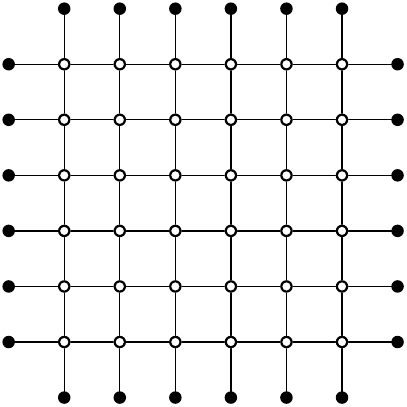}
\caption{The lower bound construction for $k=6$.}\label{fig:grid1}
\end{figure}

Consider the $k$-by-$k$ grid with leaf terminals attached in Figure~\ref{fig:grid1}: begin with a $k$-by-$k$ grid of non-terminals, and add the following leaf (degree-1) terminals:
 \begin{enumerate}
 \item one leaf terminal adjacent to each vertex on the top row, called the \emph{top} terminals,
 \item one leaf terminal adjacent to each vertex on the bottom row, called the \emph{bottom} terminals,
 \item one leaf terminal adjacent to each vertex on the left column, called the \emph{left} terminals, and
 \item one leaf terminal adjacent to each vertex on the right column, called the \emph{right} terminals.
 \end{enumerate}
Note that non-terminals on the corner of the grid have two terminals attached. This concludes the construction of the terminal set $T$.  
We now show that any edge cut sparsifier of $G$, even with directed edges allowed, has at least $k^2$ vertices. Since $G$ has $O(k)$ terminals, this proves the quadratic lower bound. 

\begin{lemma}\label{lem:1}
Any (undirected) edge cut sparsifier of $G$ has $\Omega(k^2)$ vertices.
\end{lemma}

\begin{figure}\centering
\includegraphics[scale=1.2]{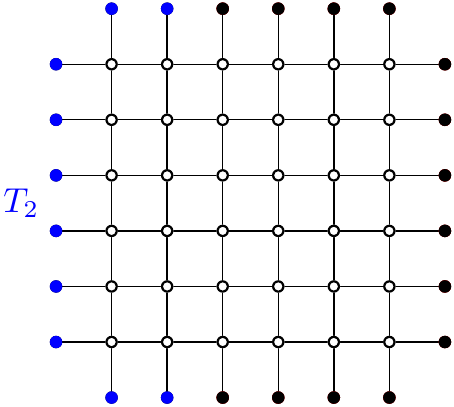}
\caption{The set $T_i$ for $i=2$.}\label{fig:grid2}
\end{figure}

For the rest of this section, we prove \Cref{lem:1}. 
For an undirected graph $G$ and a subset of vertices $S$, we define $\partial_{G}S$ as the set of edges with exactly one endpoint in $S$. 

Consider a directed sparsifier $H$ with vertex set $V'\supseteq T$. For $0\le i\le k$, let $T_i\subseteq T$ be following set of terminals: all of the left terminals, plus the first $i$ top and bottom terminals, counting from the left (see Figure~\ref{fig:grid2}). Let $C_i\subseteq V'$ be $T_i$'s side of the mincut between $T_i,T\setminus T_i$ in the sparsifier $H$, which must have $|\partial C_i| = |\partial(V'\setminus C_i)| = k$, since $\emincut_G(T_i,T\setminus T_i)$ is $k$ and $H$ is a sparsifier of $G$. Similarly, let $T'_i\subseteq T$ be all the top terminals, plus the first $i$ left and right terminals, counting from the top, and let $R_i\subseteq V'$ be $T'_i$'s side of the mincut between $T'_i,T\setminus T'_i$ in $H$, so that $|\partial R_i|=|\partial(V'\setminus R_i)|=k$.

\begin{claim}
Without loss of generality, we may assume that $C_0\subseteq C_1\subseteq \cdots\subseteq C_k$ and $R_0\subseteq R_1\subseteq\cdots\subseteq R_k$.
\end{claim}
\begin{proof}
We only prove the statement for $C_i$; the one for $R_i$ follows by symmetry of $G$. Suppose for contradiction that $C_i\setminus C_{i+1}\ne\emptyset$. By submodularity,
\[ |\partial C_i|+|\partial C_{i+1}| \ge |\partial(C_i\cap C_{i+1})|+|\partial(C_i\cup C_{i+1})| .\]
Since $C_i\cap C_{i+1}$ is a $(T_i,T\setminus T_i)$-cut in $H$ and $C_i\cup C_{i+1}$ is a $(T_{i+1},T\setminus T_{i+1})$-cut in $H$, their cut values $\partial(C_i\cap C_{i+1})$ and $\partial(C_i\cup C_{i+1})$ are at least $k$. Therefore,
\[ k+k = |\partial C_i|+|\partial C_{i+1}| \ge |\partial(C_i\cap C_{i+1})|+|\partial(C_i\cup C_{i+1})| \ge k+k ,\]
so the inequality must be an equality. It follows that we can replace $C_i$ with $C_i\cap C_{i+1}$, which is still a $(T_i,T\setminus T_{i+1})$-cut in $H$, and we can also replace $C_{i+1}$ with $C_i\cup C_{i+1}$. While there exists an $i$ such that $C_i\setminus C_{i+1}\ne\emptyset$, we perform the replacement; this can only happen a finite number of times, since the quantity $\sum_{i=0}^k|C_i|^2$ increases by at least $1$ each time and has an upper limit.
\end{proof}

For each $1\le i,j<k$, define $S_{i,j}\subseteq V'$ as $S_{i,j}=(C_{i+1}\setminus C_i)\cap(R_{j+1}\setminus R_j)$; see Figure~\ref{fig:grid3}. Our goal is to prove that $S_{i,j}\ne\emptyset$ for all $i,j$; since the sets are disjoint over all $1\le i,j\le k$, this implies the $k^2$ bound.
\begin{claim}\label{clm:3}
There are no edges between $S_{i,j}$ and $S_{i',j'}$ for $i\ne i'$ and $j\ne j'$.
\end{claim}
\begin{proof}
First, consider some $1\le i,i',j,j'\le k$ where  $i<i'$ and $j<j'$ (see Figure~\ref{fig:grid3}). By submodularity on the sets $C_{i+1}$ and $R_{j'+1}$,
\[ |\partial C_{i+1}|+|\partial R_{j'+1}| \ge |\partial(C_{i+1}\cap R_{j'+1})|+|\partial(C_{i+1}\cup R_{j'+1})| .\]
Since $\partial (C_{i+1}\cap R_{j'+1})$ is a $(T_{i+1}\cap T'_{j'+1}, T\setminus (T_{i+1}\cap T'_{j'+1}))$-cut, its value is at least
\[ \emincut_G(T_{i+1}\cap T'_{j'+1}, T\setminus (T_{i+1}\cap T'_{j'+1})) = (i+1)+(j'+1) .\]
Similarly, $|\partial(C_{i+1}\cup R_{j'+1})|\ge k-(i+1)+k-(j'+1)$. It follows that
\begin{align*}
 k+k= |\partial C_{i+1}|+|\partial R_{j'+1}| &\ge |\partial(C_{i+1}\cap R_{j'+1})|+|\partial(C_{i+1}\cup R_{j'+1})| \\&\ge (i+1)+(j'+1) +k-(i+1)+k-(j'+1) \\&= k+k, 
\end{align*}
so the inequality must be an equality. It follows that there are no edges between $C_{i+1}\setminus R_{j'+1}$ and $R_{j'+1}\setminus C_{i+1}$, since those edges would make the submodularity inequality strict. In particular, there are no edges between $S_{i,j}$ and $S_{i',j'}$.

Finally, the $i<i'$ and $j<j'$ assumptions can be removed essentially by symmetry, replacing $C_{i+1}$ by $V'\setminus C_{i+1}$ or $R_{j'+1}$ by $V'\setminus R_{j'+1}$ (or both).
\end{proof}

\begin{figure}[H]\centering
\begin{subfigure}[t]{0.45\textwidth}\centering
\includegraphics[scale=1.0]{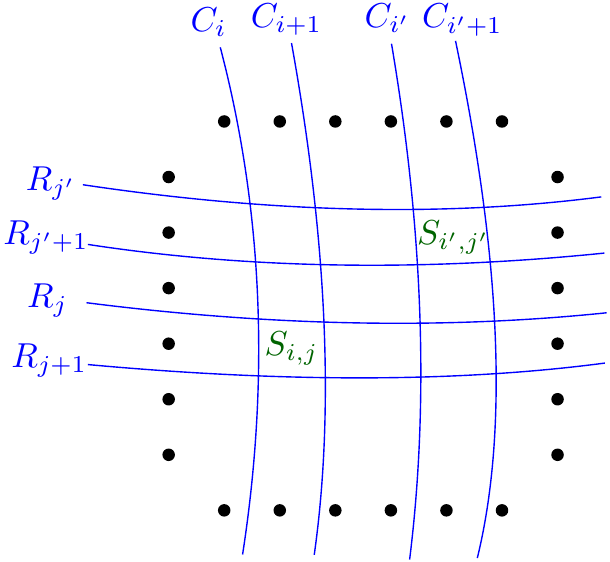}
\caption{The setting for the proof of \Cref{clm:3}.}\label{fig:grid3}
\begin{minipage}{0.1cm}
\vfill
\end{minipage}
\end{subfigure}
\hfill
\begin{subfigure}[t]{0.45\textwidth}\centering
\includegraphics[scale=1.0]{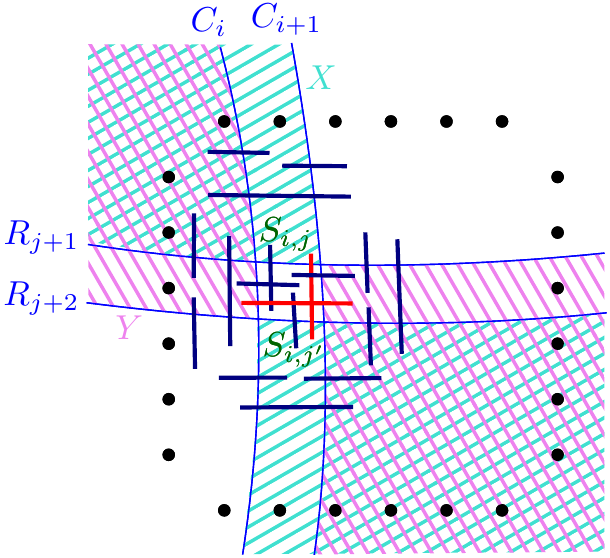}
\caption{The definitions of $X$ and $Y$, and the different types of edges that are allowed by \Cref{clm:3} and leave or enter the sets $R_{j+1}$, $R_{j+2}$, $C_i$, $C_{i+1}$. 
The long edges marked in red make the inequality $|\partial_H X| + |\partial_H Y| \le |\partial_H R_{j+1}| + |\partial_H R_{j+2}| + |\partial_H C_{i}| + |\partial_H C_{i+1}|$ strict.}\label{fig:grid5}
\end{subfigure}

\begin{subfigure}[b]{0.45\textwidth} \centering
\includegraphics[scale=1.0]{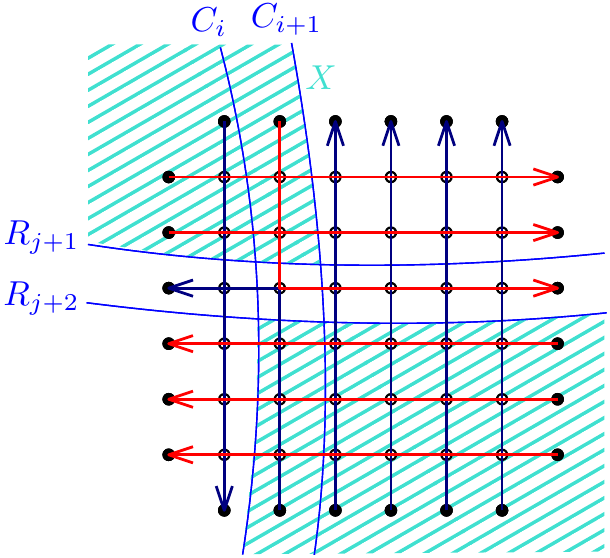}
\caption{The cut and flow that establishes that $\emincut_G(X\cap T,T\setminus X)=2k$. The flow is drawn in red and dark blue.}\label{fig:grid6}
\end{subfigure}
\hfill
\begin{subfigure}[b]{0.45\textwidth}  \centering
\includegraphics[scale=1.0]{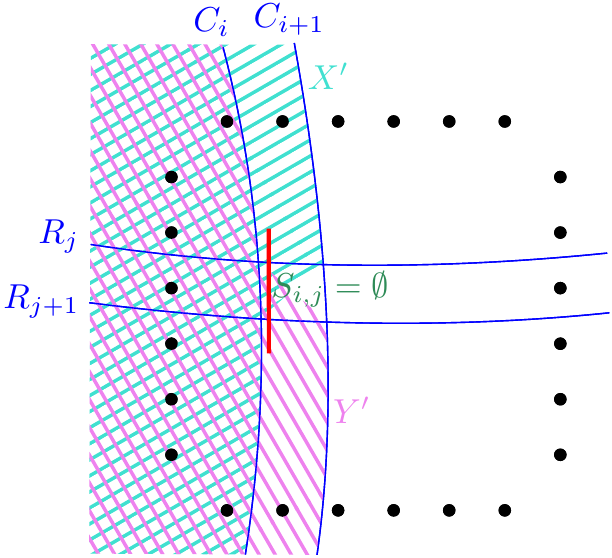}
\caption{The definitions of $X'$ and $Y'$. The red edge makes the submodularity inequality strict.}\label{fig:grid7}
\end{subfigure}
\caption{Figures for Lemma~\ref{lem:1}}
\label{fig:combo}
\end{figure}
\begin{claim}\label{clm:4}
There are no edges between $S_{i,j}$ and $S_{i,j'}$ for $|j-j'|\ge2$. Similarly, there are no edges between $S_{i,j}$ and $S_{i',j}$ for $|i-i'|\ge2$.
\end{claim}
\begin{proof}
We first prove the inequality for $S_{i,j}$ and $S_{i,j'}$ for $j'\ge j+2$.
As shown in Figure~\ref{fig:grid5}, define $X=(R_{j+1}\cap C_{i+1})\cup((V\setminus R_{j+2})\cap(V\setminus C_i))$ and $Y=(R_{j+2}\cap C_i)\cup((V'\setminus R_{j+1})\cap(V'\setminus C_{i+1})$. We now claim the inequality
\begin{gather} |\partial_H X| + |\partial_H Y| \le |\partial_H R_{j+1}| + |\partial_H R_{j+2}| + |\partial_H C_{i}| + |\partial_H C_{i+1}| \label{eqn:1}\end{gather}
by examining each type of edge in Figure~\ref{fig:grid5} and its contribution to both sides of the inequality:
 \begin{enumerate}
 \item Each edge between $C_{i}\cap R_{j+1}$ and $(C_{i+1}\sm C_i)\cap R_{j+1}$ contributes $1$ to $|\partial_HY|$ and $1$ to $|\partial_HC_{i}|$.
 \item Each edge between $C_{i+1}\cap R_{j+1}$ and $(V\sm C_{i+1})\cap R_{j+1}$ contributes $1$ to $|\partial_HX|$ and $1$ to $|\partial_HC_{i+1}|$.
 \item Each edge between $C_{i}\cap R_{j+1}$ and $(V\sm C_{i+1})\cap R_{j+1}$ contributes $1$ to $|\partial_HX|$ and $|\partial_HY|$ and $1$ to $|\partial_HC_{i}|$ and $|\partial_HC_{i+1}|$.
 \item Each edge between $C_{i}\cap R_{j+1}$ and $C_i\cap(R_{j+2}\sm R_{j+1})$ contributes $1$ to $|\partial_HX|$ and $1$ to $|\partial_HR_{j+1}|$.
 \item Each edge between $C_{i}\cap R_{j+2}$ and $C_i\cap(V\sm R_{j+2})$ contributes $1$ to $|\partial_HY|$ and $1$ to $|\partial_HR_{j+1}|$.
 \item Each edge between $C_{i}\cap R_{j+1}$ and $C_i\cap(V\sm R_{j+2})$ contributes $1$ to $|\partial_HX|$ and $|\partial_HY|$ and $1$ to $|\partial_HR_{j+1}|$ and $|\partial_HR_{j+2}|$.
 \item Each edge between $C_{i}\cap(V\sm R_{j+2})$ and $(C_{i+1}\sm C_i)\cap(V\sm R_{j+2})$ contributes $1$ to $|\partial_HX|$ and $1$ to $|\partial_HC_{i}|$.
 \item Each edge between $C_{i+1}\cap(V\sm R_{j+2})$ and $(V\sm C_{i+1})\cap(V\sm R_{j+2})$ contributes $1$ to $|\partial_HY|$ and $1$ to $|\partial_HC_{i+1}|$.
 \item Each edge between $C_{i}\cap(V\sm R_{j+2})$ and $(V\sm C_{i+1})\cap(V\sm R_{j+2})$ contributes $1$ to $|\partial_HX|$ and $|\partial_HY|$ and $1$ to $|\partial_HC_{i}|$ and $|\partial_HC_{i+1}|$.
 \item Each edge between $(V\sm C_{i+1})\cap R_{j+1}$ and $(V\sm C_{i+1})\cap(R_{j+2}\sm R_{j+1})$ contributes $1$ to $|\partial_HY|$ and $1$ to $|\partial_HR_{j+1}|$.
 \item Each edge between $(V\sm C_{i+1})\cap R_{j+2}$ and $(V\sm C_{i+1})\cap(V\sm R_{j+2})$ contributes $1$ to $|\partial_HX|$ and $1$ to $|\partial_HR_{j+1}|$.
 \item Each edge between $(V\sm C_{i+1})\cap R_{j+1}$ and $(V\sm C_{i+1})\cap(V\sm R_{j+2})$ contributes $1$ to $|\partial_HX|$ and $|\partial_HY|$ and $1$ to $|\partial_HR_{j+1}|$ and $|\partial_HR_{j+2}|$.
 \item Each edge between $C_{i}\cap(R_{j+2}\sm R_{j+1})$ and $(C_{i+1}\sm C_i)\cap(R_{j+2}\sm R_{j+1})$ contributes $1$ to $|\partial_HY|$ and $1$ to $|\partial_HC_i|$.
 \item Each edge between $(C_{i+1}\sm C_i)\cap(R_{j+2}\sm R_{j+1})$ and $(V\sm C_{i+1})\cap(R_{j+2}\sm R_{j+1})$ contributes $1$ to $|\partial_HY|$ and $1$ to $|\partial_HC_{i+1}|$.
 \item Each edge between $ C_i\cap(R_{j+2}\sm R_{j+1})$ and $(V\sm C_{i+1})\cap(R_{j+2}\sm R_{j+1})$ contributes $1$ to $|\partial_HC_i|$ and $1$ to $|\partial_HC_{i+1}|$.
 \item Each edge between $(C_{i+1}\sm C_i)\cap R_{j+1}$ and $(C_{i+1}\sm C_i)\cap(R_{j+2}\sm R_{j+1})$ contributes $1$ to $|\partial_HX|$ and $1$ to $|\partial_HR_{j+1}|$.
 \item Each edge between $(C_{i+1}\sm C_i)\cap (R_{j+2}\sm R_{j+1})$ and $(C_{i+1}\sm C_i)\cap(V\sm R_{j+2})$ contributes $1$ to $|\partial_HX|$ and $1$ to $|\partial_HR_{j+2}|$.
 \item Each edge between $(C_{i+1}\sm C_i)\cap  R_{j+1}$ and $(C_{i+1}\sm C_i)\cap(V\sm R_{j+2})$ contributes $1$ to $|\partial_HR_{j+1}|$ and $1$ to $|\partial_HR_{j+2}|$.
 \end{enumerate}
All of the above types of edges contribute the same to both sides of (\ref{eqn:1}) except those of type~15 and~18, namely those indicated by long red edges in Figure~\ref{fig:grid5}. We call such edges \emph{red}. Then, the inequality (\ref{eqn:1}) is strict if any only if red edges are present.

Since all edges between $S_{i,j}$ and $S_{i,j'}$ are red, it suffices to show that (\ref{eqn:1}) is actually an equality, which would exclude all red edges and hence all edges  between $S_{i,j}$ and $S_{i,j'}$ as well.
Observe that $\emincut_G(X\cap T,T\setminus X)=2k$, as seen in Figure~\ref{fig:grid6}, which shows a cut and a flow both of value $2k$. Similarly, $\emincut_G(Y\cap T,T\setminus Y)=2k$. Since $H$ is a sparsifier of $G$, we must have
\begin{align*}
 2k+2k = |\dvec\partial_HX| + |\dvec\partial_HY| &\le |\dvec\partial_HR_{j+1}| + |\dvec\partial_HR_{j+2}| + |\dvec\partial_HC_{i}| + |\dvec\partial_HC_{i+1}| \\&= k+k+k+k .
\end{align*}
In other words, the inequality is tight, as desired.

The case for $S_{i,j}$ and $S_{i',j}$ for $i'\ge i+2$ is similar. Note that edges between $S_{i,j}$ and $S_{i',j}$ correspond to the red horizontal edges in Figure~\ref{fig:grid5} (for different values of $i,j$), which we have already shown do not exist.
\end{proof}

\begin{claim}
For each $1\le i,j\le k$, we have $S_{i,j}\ne\emptyset$.
\end{claim}
\begin{proof}
Suppose for contradiction that $S_{i,j}=\emptyset$ for some $1\le i,j\le k$. As shown in Figure~\ref{fig:grid7}, define the sets $X'=C_i\cup (C_{i+1}\cap R_j)$ and $Y'=C_i\cup(C_{i+1}\setminus R_{j+1})$. We have $X'\cap Y'=C_i$, and since $S_{i,j}=\emptyset$ by assumption, we also have $X'\cup Y'=C_{i+1}$. By submodularity,
\[ |\dvec\partial_HX'| + |\dvec\partial_HY'| \ge |\dvec\partial_H(X'\cap Y')| + |\dvec\partial_H(X'\cup Y')| =  |\dvec\partial_HC_i| + |\dvec\partial_HC_{i+1}|. \]
Moreover, the inequality is tight because the only types of edges that can make the inequality strict (marked red in Figure~\ref{fig:grid7}) are prohibited by \Cref{clm:4}. Since $H$ is a sparsifier of $G$,
\[ \emincut_G(X'\cap T,T\setminus X')+\emincut_G(Y'\cap T,T\setminus Y') \le |\dvec\partial_HX'| + |\dvec\partial_HY'| =  |\dvec\partial_HC_i| + |\dvec\partial_HC_{i+1}|= 2k.\]
However, it is not hard to see that $\emincut_G(X'\cap T,T\setminus X')=\emincut_G(Y'\cap T,T\setminus Y')=k+1$, a contradiction.
\end{proof}

\section{Conclusions}
We showed that every unweighted, directed acyclic graph $G$ with $k$ terminals
admits a vertex cut sparsifier $H$ with $O(k^2)$ vertices, assuming that the
terminals are deletable.  This improves the previous result by \KW of $O(k^3)$ 
vertices, for general directed graphs~\cite{KratschW20}.  Furthermore, the
sparsifier can be computed in near-linear time in the size of $G$, specifically
in time $O((m+n)k^{O(1)})$ plus $O((m+n)k^{O(1)})$ field operations over a finite field with 
entries of bitlength $O(k \log n)$, where $n=|V(G)|$ and $m=|E(G)|$.  This improves over previous work~\cite{KratschW20},
whose time complexity was not explicitly given but is at least $O(n^{\omega+1})$
due to the repeated construction of a representation of a gammoid.

Furthermore, we showed that $\Omega(k^2)$ vertices in a sparsifier may be required, both for
vertex cuts in DAGs and for the seemingly simpler setting of undirected edge cuts.
However, we leave it open whether such a bound applies to the mixed setting, where
we want to preserve undirected edge cuts in the input graph $G$ but allow
the sparsifier to be a directed graph. 

More importantly, we leave open the question of whether vertex cut
sparsifiers of $O(k^2)$ vertices exist for general directed graphs.
We conjecture that $O(k^2)$ is the correct bound for general directed graphs,
but we were not able to find a proof.


\bibliographystyle{plainurl}
\bibliography{bib}


\section*{Appendix}
\paragraph*{Proof of Lemma~\ref{lem:quick-repr}}
\begin{lemma*}
Given a directed acyclic graph $G = (V, E)$, a set $S\subseteq V$, and $\varepsilon > 0$, with $|V|=n$, $|E|=m$ and $|S|=k$, a representation of the gammoid on $S$ of dimension $k$ over a finite field with entries of bit length $\ell=O(k \log n + \log(1/\varepsilon))$ can be constructed in randomized time $\tilde O((n+m) k\ell)$ with one-sided error at most $\varepsilon$, where $\tilde O$ hides factors logarithmic in $\ell$. 
\end{lemma*}
\begin{proof}
We review the construction of Mason~\cite{mason72gammoids}. Associate a variable
$x_{uv}$ to every edge $(u,v) \in E$, where all variables $x_{uv}$
are formally independent. For two vertices $u, v \in V$,
define the path polynomial
\[
  P(u,v) = \sum_{Q \colon u \leadsto v} \prod_{(u,v) \in E(Q)} x_{uv}
\]
where $Q$ ranges over all directed paths from $u$ to $v$ in $G$.
Define the matrix $M$ with rows indexed by $S$ and columns by $V$,
where for $u \in S$, $v \in V$ we have $M(u,v)=P(u,v)$.
We claim that $M$ is a representation of the gammoid on $S$. 

Indeed, on the one hand, let $T \subseteq V$ be a basis of the gammoid.
By definition, there is a vertex-disjoint flow linking $S$ to $T$.
Instantiate the variables $x_{uv}$ by letting $x_{uv}=1$ for every edge
used in one of these paths, and $x_{uv}=0$ otherwise.  Under this
evaluation $M[S,T]$ is a permutation matrix, hence independent,
which shows that $\det M[S,T] \not \equiv 0$. On the other hand, 
let $T \subseteq V$ and let $C \subseteq V$ be an $(S,T)$-min cut,
$|C|<|T|$. Then $M[S,T]$ factors as $M[S,T] = M[S,C] \cdot M'[C,T]$ for a matrix $M'[C,T]$, hence the rank of
$M[S,T]$ is at most $|C|$.   Here, $M'[C,T]$ is defined as $M$, except that vertices of $C$ have been turned into sources.

To get a representation over a finite field $\mathbb{F}$, we pick a
sufficiently large field $\mathbb{F}$ and replace every variable $x_{uv}$ by
a random value from $\mathbb{F}$.  For the success probability, note that
any dependent set in $M$ remains dependent after such a replacement. 
Therefore, it is enough to consider the probability that
$\det M[S,B] \neq 0$ for every basis $B$ of the gammoid. 
For this, we observe that the entries $M(u,v)$ are polynomials of
degree at most $n$, hence $\det M[S,B]$ has degree at most $nk$.
Furthermore, the number of bases is at most $\binom{n}{k} \leq n^k$.
Let $\mathbb{F}=GF(2^\ell)$ where $2^\ell > (1/\varepsilon)(nk)n^k$, i.e., $\ell = \Theta(k \log n + \log (1/\varepsilon))$. 
By Schwartz-Zippel, the probability that $\det M[S,B]=0$ for a given
basis $B$ is at most $nk/|\mathbb{F}| \leq \varepsilon n^{-k}$, hence by the
union bound the probability that this occurs for at least one basis $B$
is at most $\varepsilon$. We note that field arithmetic over $\mathbb{F}$ can
be performed in time $\tilde O(\ell)$.

It remains to evaluate the vectors $R_v=(P(s,v))_{s \in S}$ for $v \in V$
quickly. For simplicity, we assume w.l.o.g.\ that the vertices $s \in S$ 
are sources in $G$; this can be achieved by introducing a new vertex
$s'$ for every $s \in S$, with a single edge $(s',s)$, and
replacing $s$ by $s'$ in $S$.  Note that this does not change the
resulting gammoid.  Let $V=\{v_1,\ldots,v_n\}$ where $(v_1, \ldots, v_n)$ 
is a topological ordering of $G$, starting with the vertices of $S$;
this can be computed in time $O(m+n)$ by standard methods.  Note
that $P(s,s)=1$ for $s \in S$ and $P(s, \overline{s}) = 0$ for $s\ne \overline{s}\in S$, hence the vectors $R_s$, $s \in S$ 
are unit vectors making up the standard basis for $\mathbb{F}^k$. For all
other vertices $v \in V$, note
\[
  P(s,v) = \sum_{u \in N^-(v)} P(s,u)x_{uv},
\]
where $P(s,u)=R_u(s)$ has already been computed due to the
topological ordering. Hence $R_v$ can be computed using
$O(kd^-(v))$ field operations from the previously computed vectors. 
Performing this across all variables $v_1, \ldots, v_n$ uses $O((n+m)k)$
field operations, hence the total running time is bounded by
$\tilde O((n+m)k \ell)$, as stated.

Finally, we note that a similar representation result can be shown for
general digraphs, if the paths in the path polynomial are replaced by
walks.  However, in a general digraph there is no known way of
evaluating the result faster than matrix multiplication time, so this
does not lead to any algorithmic improvements.
\end{proof}

\paragraph*{Proof of Lemma~\ref{lem:ess}}
\begin{lemma*}
Let $v$ be an essential vertex with respect to $X\subseteq S, Y\subseteq T$. Let $C$ be any minimum vertex cut between $X, Y$. Then $(C, v)$ is saturated by both $X$ and $Y$.
\end{lemma*}
\begin{proof}
We slightly modify our graph. Add a vertex $v'$ to $G$, and add edges such that $v'$ has the same in-neighbors and out-neighbors as $v$. Denote the new graph as $G'$, we show that $C' = C\cup \{v'\}$ is a min-cut between $X, Y$ in $G'$. This lemma then follows from Menger's Theorem.

Assume for the sake of contradiction that $C'$ is not a min-cut between $X, Y$, and let $D$ be a min-cut between $X, Y$ where $|D| < |C'| = |C| + 1$. We consider a few cases.
\begin{enumerate}[label = (\alph*)]
\item $v\in D$ and $v'\notin D$. Since $D$ is a min-cut, by Menger's theorem there exists a path $P$ from $X$ to $Y$ such $P\cap D = \{v\}$. If we replace $v$ by $v'$ in $P$, we obtain a valid path $P'$ from $X$ to $Y$ that does not cross $D$, which is a contradiction to the fact that $D$ is a cut.

\item $v'\in D$ and $v\notin D$. The same argument from the previous case works. 

\item $v, v' \notin D$. Then $D$ is a valid cut in $G$, and $|D|\le |C|$. Therefore $D$ is a min-cut in $G$ that does not contain $v$, which is a contradiction. 

\item $v, v' \in D$. Then consider $D' = D\setminus \{v'\}$. $D'$ is a valid cut in $G$, and $|D'| \le |C| - 1$, which contradicts the fact that $C$ is a min-cut in $G$.
\end{enumerate}
Therefore we conclude that $C'$ is a min-cut in $G'$, which completes the proof.
\end{proof}

\paragraph*{Proof of Lemma~\ref{lem:cl}}
\begin{lemma*}
If $v\in V\setminus(S\cup T)$ is not an essential vertex, then $\cl_v(G)$ is a vertex cut sparsifier of $G$. 
\end{lemma*}
\begin{proof}
Let $H = \cl_v(G)$, it suffices for us to show that for all $X\subseteq S, Y\subseteq T$, 
\[
\mc_G(X, Y) = \mc_H(X, Y).
\]
We first show $\mc_G(X, Y) \le \mc_H(X, Y)$. Let $C$ be a min-cut between $X, Y$ in $H$, we show that $C$ is a valid vertex cut in $G$. Suppose not, then there exists a path $P$ in $G$ from $X$ to $Y$ that does not intersect with $C$. If $v\notin P$, then $P$ is also a path in $H$, which contradicts the fact that $C$ is a cut. If $v\in P$, let $u$ and $w$ be $v$'s predecessor and successor in $P$. According to the closure operation on $v$, there is an edge $uw$ in $H$. Therefore we again have a path from $X$ to $Y$ in $H$ not intersecting $C$, contradiction. 

Next we show $\mc_G(X, Y) \ge \mc_H(X, Y)$. Since $v$ is not essential, there exists a min-cut $C$ between $X, Y$ such that $v\notin C$. We show that $C$ is a cut in $H$. Suppose not, then there exists a path $P$ in $H$ from $X$ to $Y$ not crossing $C$. For every edge $uw$ on $P$ that is created by the closure operation, replace it by the path $u,v,w$ and denote the new walk $P'$. After truncating all the cycles of $P'$, we obtain a path from $X$ to $Y$ in $G$ not crossing $C$, contradiction. We conclude that  $\mc_G(X, Y) = \mc_H(X, Y)$.
\end{proof}

\paragraph*{Proof of Lemma~\ref{lem:close}}
\begin{lemma*}
Let $C$ be a vertex cut for $X\subseteq S, Y\subseteq T$ that is closest to $X$ ( resp. $Y$), then for all vertices $v\in C$, $(C, v)$ is saturated by $X$ (resp. $Y$).
\end{lemma*}
\begin{proof}
We prove this lemma for cuts closest to $X$, as the other case is symmetrical. Assume for the sake of contradiction there exists $v\in C$ such that $(C, v)$ is not saturated by $X$. Add a sink-only copy of $v$ into our digraph $G$ and call the new graph $G'$. By duality there must exist a cut $C'$ of size $|C|$ between $X$ and $C\cup \{v'\}$ in $G'$. We now consider a few cases.

First, note that we cannot have $v, v'$ both in $C'$, as otherwise $C'\sm \{v'\}$ would be a smaller $(X, Y)$ cut in $G$. Now if $v\in C'$ and $v'\notin C'$, then by duality there is a path from $X$ to $v$ that does not intersect $C'\sm\{v\}$, which can be redirected to $v'$, contradicting the fact that $C'$ is a cut. Similarly, if $v'\in C'$ and $v\notin C$, we have a contradiction. Therefore $v, v'\notin C'$, which implies $C$ is not the unique $(X, C)$ min-cut, contradiction. 
\end{proof}

\paragraph*{Proof of Lemma~\ref{lemma:kw-quadratic}}

\begin{lemma*}
    Let $S$ and $T$ be two vertex sets of size $2k$. Enumerate them as 
    \[  
    S = \{v_1, v_1', v_2, v_2', \cdots, v_k, v_k'\}, \\
    T = \{u_1, u_1', u_2, u_2', \cdots, u_k, u_k'\}.
    \]
For each $i, j\in [k]$, create a vertex $w_{i,j}$ and add edges from $v_i, v_i'$ to $w_{i,j}$, and from $w_{i,j}$ to $u_j, u_j'$. 
    Any vertex cut sparsifier for the resulting graph requires $\Omega(k^2)$ vertices.
    \end{lemma*}
\begin{proof}
    Let $\lambda_G$ be the cut function of $G$, i.e,. for $A \subseteq S$ and $B \subseteq T$, let $\lambda_G(A,B)$ be the size of a minimum $(A,B)$-vertex cut in $G$.  Let $H$ be a vertex cut sparsifier of $G$ for $S$ and $T$.  For $i, j \in [k]$, let $C_{i,j}$ be some min-cut in $H$ between $A_i:=\{u_i,u_i'\}$ and $B_j:=\{v_j,v_j'\}$.  We claim that all sets $C_{i,j}$ are disjoint. 
    
    Indeed, note that $|C_{i,j}|=\lambda_G(A_i,B_j)=1$, hence it suffices to show that $C_{i,j} \neq C_{i',j'}$ for $(i,j) \neq (i',j')$.  First, assume to the contrary that $C_{i,j}=C_{i,j'}$ for some $j \neq j'$.  Let $H'=H-C_{i,j}$.  Then $H'$ has no path from $A_i$ to $B_j \cup B_{j'}$, and $|C_{i,j}|=1$, but $\lambda_G(A_i, B_j \cup B_{j'})=2$, contradicting that $H$ is a cut sparsifier.  The symmetric argument holds for $C_{i,j}=C_{i',j}$ for some $i \neq i'$.  Finally, assume $C_{i,j}=C_{i',j'}$ for some $i \neq i'$, $j \neq j'$.  Then as above $C_{i,j} \cup C_{i,j'} \cup C_{i',j}$ is a cut in $H$ between $A_i \cup A_{i'}$ and $B_j \cup B_{j'}$, contradicting $\lambda_G(A_i \cup A_i', B_j \cup B_j')=4$.
\end{proof}

\end{document}